\documentclass[11pt]{amsart}
\usepackage{booktabs}

\usepackage{amsmath,amsfonts,amsthm,mathrsfs,xspace,graphicx}
\usepackage[
colorlinks,bookmarks=true, citecolor=blue!90!black, urlcolor=blue!90!black,linkcolor=blue!90!black, pdfencoding=auto, psdextra]{hyperref}
\hypersetup{pdfauthor={Name}}
\usepackage{endnotes}
\usepackage{color}
\usepackage{float}
\usepackage{subcaption}
\usepackage{xprintlen}
\usepackage{array}
\newcolumntype{C}[1]{>{\centering\arraybackslash}p{#1}}
\newcolumntype{L}[1]{>{\raggedright\arraybackslash}p{#1}}

\usepackage{xcolor}
\definecolor{since}{rgb}{0.5,0.5,0.5}
\definecolor{newred}{HTML}{ff382e}
\definecolor{newgreen}{HTML}{549641}
\definecolor{newblue}{HTML}{4c4cfc}
\definecolor{neworange}{HTML}{c98702}

\usepackage{mdframed}
\usepackage{framed}
\usepackage{bbm}
\usepackage{suffix} 
\usepackage{tabularx}
\usepackage{makecell}
\usepackage{amssymb,latexsym}
\usepackage[capitalize]{cleveref}
\usepackage{enumitem}
\usepackage{tikz}
\usetikzlibrary{quantikz2}

\usepackage{multirow}
\usepackage{endnotes}
\usepackage{amssymb,latexsym}
\usepackage{enumitem}
\usepackage{float}
\usepackage{setspace}
\usepackage{soul}
\usepackage[section]{placeins}
\usepackage{thmtools}
\usepackage{thm-restate}
\usepackage{stmaryrd}
\usepackage{esvect}
\usepackage{mathtools, physics}
\usepackage{complexity}
\usepackage{verbatim}

\usepackage{pdflscape}
\usepackage{adjustbox}

\usepackage{algorithm}
\usepackage{algpseudocode}

\makeatletter
\renewcommand*\env@matrix[1][*\c@MaxMatrixCols c]{%
  \hskip -\arraycolsep
  \let\@ifnextchar\new@ifnextchar
  \array{#1}}
\makeatother

\newtheorem{theorem}{Theorem}[section] 
\newtheorem*{theorem*}{Theorem}
\newtheorem{proposition}[theorem]{Proposition}
\newtheorem*{proposition*}{Proposition}
\newtheorem{lemma}[theorem]{Lemma}
\newtheorem*{lemma*}{Lemma}
\newtheorem{corollary}[theorem]{Corollary}
\newtheorem*{corollary*}{Corollary}
\newtheorem{conjecture}{Conjecture}
\newtheorem*{conjecture*}{Conjecture}

\theoremstyle{definition}

\newtheorem*{fact*}{Fact}
\crefname{fact}{Fact}{Facts}
\newtheorem{claim}[theorem]{Claim}
\crefname{claim}{Claim}{Claims}
\newtheorem{definition}[theorem]{Definition}
\newtheorem*{definition*}{Definition}

\DeclareRobustCommand{\EndDef}{%
}

\theoremstyle{remark}
\newtheorem{remark}[theorem]{Remark}
\AtEndEnvironment{remark}{\EndDef}
\AtEndEnvironment{definition}{\EndDef}
\newtheorem*{remark*}{Remark}
\newtheorem{example}{Example}

\newclass{\QPCP}{QPCP}
\newclass{\QCPCP}{QCPCP}
\newclass{\QCMAcomp}{QCMA-complete}
\newclass{\sharpP}{\#P}

\DeclareMathOperator\CSS{CSS}

\DeclareMathOperator\Span{Span}

\DeclareMathOperator\eye{\mathbb{I}}

\DeclareMathOperator{\supp}{supp}

\DeclareMathOperator\dist{dist}

\usepackage{accents}

\newcommand{\br}[1]{\ensuremath{\left\{{#1}\right\}}}

\newcommand{\bigmid}{\;\big\vert\;}
\newcommand{\Bigmid}{\;\,\Big\vert\,\;}

\newcommand{\standard}[1]{{\left\langle{#1}\right\rangle}}
\DeclareRobustCommand{\euler}{\genfrac{\langle}{\rangle}{0pt}{}}
\newcommand{\indicator}[1]{\mathbbold{1}_{#1}}

\newclass{\QC}{QC}
\newcommand{\Coxeter}[1]{\ensuremath{\C_{W}({#1})}}
\newcommand{\QCoxeter}[1]{\ensuremath{\QC_{W}({#1})}}
\newcommand{\coxeter}[3]{\ensuremath{\C_{{#1}}({#3})}}
\newcommand{\Qcoxeter}[3]{\ensuremath{\QC_{{#1}}({#3})}}

\newcommand{\ext}[1]{\ensuremath{\mathcal{E}_{#1}}}

\newcommand{\CC}{\ensuremath{\mathbb{C}}}
\newcommand{\NN}{\ensuremath{\mathbb{N}}}
\newcommand{\FF}{\ensuremath{\mathbb{F}}}

\newcommand{\RR}{\ensuremath{\mathbb{R}}}
\newcommand{\ZZ}{\ensuremath{\mathbb{Z}}}

\newcommand{\mcB}{\ensuremath{\mathcal{B}}} 
\newcommand{\mcC}{\ensuremath{\mathcal{C}}}

\newcommand{\mcJ}{\ensuremath{\mathcal{J}}}

\newcommand{\mcP}{\ensuremath{\mathcal{P}}}
\newcommand{\mcQ}{\ensuremath{\mathcal{Q}}}
\newcommand{\mcR}{\ensuremath{\mathcal{R}}}
\newcommand{\mcS}{\ensuremath{\mathcal{S}}}

\DeclareMathAlphabet{\mathbbold}{U}{bbold}{m}{n}

\newcommand\restr[2]{{
  \left.\kern-\nulldelimiterspace 
  #1 
  \right|_{#2} 
  }}
  
\usepackage[margin=1.23in]{geometry}
\usepackage{enumitem,diagbox}
\usepackage{stfloats,nicefrac}

\usepackage{cite}
\usepackage[normalem]{ulem}
\newcommand\redout{\bgroup\markoverwith{\textcolor{red}{\rule[0.5ex]{2pt}{0.8pt}}}\ULon}

\newcommand{\bfit}{\bfseries\itshape}

\newcommand\nnfootnote[1]{%
   \begin{NoHyper}
    \renewcommand\thefootnote{}\footnote{#1}%
    \addtocounter{footnote}{-1}%
   \end{NoHyper}
}
\makeatletter
\newcommand\footnoteref[1]{\protected@xdef\@thefnmark{\ref{#1}}\@footnotemark}
\makeatother
\allowdisplaybreaks

\usepackage{balance}

\title[Coxeter codes]{Coxeter codes: Extending the Reed--Muller family}
\author{Nolan J. Coble}\author{Alexander Barg}
\date{}

\begin{document}

\begin{abstract}
Binary Reed--Muller (RM) codes are defined via evaluations of Boolean-valued functions on $\ZZ_2^m$. We introduce a class of binary linear codes that generalizes the RM family by replacing the domain $\ZZ_2^m$ with an arbitrary finite Coxeter group. Like RM codes, this class is closed under duality, forms a nested code sequence, satisfies a multiplication property, and has asymptotic rate determined by a Gaussian distribution. Coxeter codes also give 
rise to a family of quantum codes for which transversal diagonal $Z$ rotations can perform non-trivial logic.
\end{abstract}

\maketitle

\nnfootnote{N.C. was partially supported by NSF grant DMS-2231533. A.B. was partially supported by NSF grant CCF-2330909. An extended abstract of this work appears in Proceedings of the 2025 IEEE International Symposium on Information Theory.}



\section{A New Family of Binary Codes}
Reed--Muller (RM) codes form a classic family studied for its interesting algebraic and combinatorial properties \cite{MS77,Assmus98} as well as from the perspective of information transmission \cite{YeAbbe2020,abbe2023reed}. 
They achieve Shannon capacity of the basic binary channel models such as channels with independent erasures or flip errors
\cite{Kudekar2015ReedMullerCA}, \cite{reeves2023reed}, \cite{AbbeSandon2023}. They also give rise to a family of quantum codes \cite{Steane1999} with well-understood logical operators \cite{kubica2015universal}, \cite{campbell_magic-state_2012}, \cite{rengaswamy2020optimality}, \cite{barg2024geometric}. 
Beginning with the standard definition of RM codes, we then give an equivalent combinatorial characterization that admits a natural generalization to the Coxeter code family.

Consider the binary field $\FF\coloneqq\FF_2$ and the space of Boolean functions $\br{f\colon\ZZ_2^m\rightarrow\FF}$, which can also be defined as the group algebra $\FF\ZZ_2^m$. 
Every such Boolean function can be written as an $m$-variate polynomial, and the binary RM code $RM(r,m)$ of order $r$ is defined as the set of polynomials 
of degree at most $r$.\footnote{More precisely, codewords are \emph{evaluation vectors} of these 
polynomials; throughout, we will not distinguish functions from their evaluation vectors.}
Our starting point for Coxeter codes is to note that the group $\ZZ_2^m$ admits a \emph{combinatorial structure} of the $m$-dimensional (Boolean) hypercube graph, which is composed of smaller subcubes. To make this explicit, let $S_m\coloneqq \br{e_1,\dots,e_m}$ be the set of standard generators of $\ZZ_2^m$.
For all $\ell\in\{0,\dots,m\}$, $\ell$-dimensional subcubes arise as cosets $z+\standard{J},$ where $z\in\ZZ_2^m$ and $\standard J$ is the standard subgroup spanned by an $\ell$-subset $J\subseteq S_m$.
\begin{theorem}[\!\!\!{\cite{barg2024geometric}, Fact II.3}]\label{thm: def RM}
    For $r\in \{-1,0,\dots,m\}$ the \emph{order-$r$ Reed--Muller code} $RM(r,m)$ is equal to
    $$RM(r,m)=\Span_{\FF}\br{\indicator{z+\standard{J}}\mid z\in\ZZ_2^m, J\subseteq S_m, \abs{J}=m-r}.$$
\end{theorem}
Inclusion of $r=-1$ as a possible order value deviates from the standard definition \cite[Ch.~13]{MS77},
which is limited to $0\le r\le m$. It is convenient to extend the order set to account for the duality 
within the RM code family, and this applies to all Coxeter codes.

It is well known that the codewords of minimum weight in the code $RM(r,m)$ are given by incidence vectors of $(m-r)$-flats in the affine geometry $AG(m,2)$, and that this collection of minimum-weight codewords generates the entire code \cite[Thm.13.12]{MS77}; \cref{thm: def RM} strengthens this by pointing out that subcubes, a subset of flats, are sufficient to generate the code. It is straightforward to verify that indicator functions of $(m-r)$-dimensional subcubes are degree-$r$ ``signed'' monomials $\prod_{j=1}^r y_{i_j},$ where $y_i\in\{x_i,\bar x_i\}$ for all $i$ and $\bar x\coloneqq1-x$.

While studying \emph{quantum} RM codes \cite{barg2024geometric}, we realized that many of the simple structural properties of RM codes---containment, duality, multiplication---typically viewed as deriving from the polynomial definition, likewise arise from the combinatorial structure of $\ZZ_2^m$ when viewed as a group generated by $S_m$. For example, $(m-r)$- and $(r+1)$-dimensional subcubes necessarily intersect on an even number of elements, indicative of the duality $RM(r,m)^\perp = RM(m-r-1,m)$.  
This combinatorial structure is shared by every member of a large family of well-studied groups known as \emph{Coxeter groups}.

\begin{definition}
    Let $S\coloneqq\br{s_1,\dots,s_m}$ be a set of $m$ generators. A \emph{Coxeter group} $W$ is given by a presentation
    $$
        W\coloneqq \left\langle S\mid (s_i s_j)^{M(i,j)}=1 \right\rangle,
    $$
    where 
    $M(i,i)=1$ (i.e., $s_i^2=1$) and $M(i,j)=M(j,i)\in \ZZ_{\geq 2}$. 
    The pair $(W,S)$ is called a \emph{Coxeter system} of rank $m$ and the matrix $(M(i,j))_{i,j=1}^m$ is called the {\em defining matrix}
    of the system. 
\end{definition}
Clearly, $(\ZZ_2^m,S_m)$ is a Coxeter system with $M(i,j)=2$ for all $i,j$. A classic example of a Coxeter system is the symmetric group on $m+1$ letters, $A_m\coloneqq(\mathrm{Sym}(m+1),T)$,\footnote{Not to be confused with the $(m+1)$-letter \emph{alternating group}; in the theory of Coxeter groups, the letter $A$ refers to the full symmetric group.} where $T=\br{(i\;\; i+1)\mid i\in [m]}$ is the set of adjacent transpositions. In this case, 
$M(i,i+j)=2$ for all $j\geq 0$ except $j=1$ when $M(i,i+1)=3$. A classic
visualization of this system is shown in \cref{fig: permutohedron}, and other examples are given later in \cref{fig: Am figs,fig: I2 figs}.

A Coxeter system is called {\em irreducible} if for any partition of the generators $S=S_1\sqcup S_2$ there are $s\in S_1$ and $t\in S_2$ that do not commute, and is called {\em reducible} otherwise.
This definition provides no visual interpretation of irreducibility; a more
standard definition relies on Coxeter-Dynkin diagrams \cite{BB05}, which we do not use in this paper (except in the proof of \cref{cor: distance for large r}). Finite Coxeter groups have a succinct classification, e.g., \cite[App.A.1]{BB05}, and we will assume throughout that $W$ is a finite group.

To define Coxeter codes, we need a suitable generalization of a subcube to an arbitrary Coxeter system, where, as before, $\standard J$ denotes 
the subgroup generated by a subset $J\subset S$.
\begin{definition}
    Fix a Coxeter system, $(W,S)$. A \emph{standard subgroup of 
    $W$} is a subgroup $\langle J\rangle\leq W$ where $J\subseteq S$. 
        A \emph{standard (left) coset} of $W$ is any coset of the form $R\coloneqq \sigma\standard{J}$ for $\sigma\in W$, $J\subseteq S$. The \emph{rank} of $R=\sigma\standard{J}$ is $\rank (R)\coloneqq \abs{J}$.
\end{definition}

We now construct a family of $\FF$-linear codes from a given Coxeter system $(W,S)$ of rank $m$. Consider the group algebra $\FF W\coloneqq\br{f\colon W\rightarrow\FF}$ of $\FF$-valued functions on $W$, which is a $\abs{W}$-dimensional vector space. Let $\indicator{U}\in\FF W$ denote the \emph{indicator function} of a subset $U\subseteq W$.

\begin{definition}[\sc Coxeter codes]\label{def: Coxeter codes}
    For $r\in\br{-1,\dots,m}$, the \emph{order-$r$ Coxeter code of type $(W,S)$}, denoted by $\Coxeter{r}$, is the $\FF$-linear span
    of indicator functions of standard cosets having rank $m-r$:
    $$
    \Coxeter{r} \coloneqq\Span\br{\indicator{\sigma\standard{J}}\mid \sigma\in W, J\subseteq S, \abs{J}=m-r}.
    $$

\end{definition}

\begin{remark}
    \hspace{0em}
    \begin{itemize}[leftmargin=*]
        \item The code $\Coxeter{r}$ depends on the particular choice of $S$; we suppress this dependence in the notation for simplicity.
        \item For $\ZZ_2^m$ with its standard generating set, the order-$r$ Coxeter code of type $(\ZZ_2^m,S_m)$ is the code $RM(r,m)$ by \cref{thm: def RM}.
        \item For every Coxeter system: $\Coxeter{-1}=0^{\abs{W}}$ is a trivial code (given by an empty generating set), $\Coxeter{0}$ is a repetition code, $\Coxeter{m-1}$ is a single parity-check code, and $\Coxeter{m}=\FF W$ is the entire vector space. \hfill$\lhd$
    \end{itemize}
\end{remark}


\begin{figure}[t]
    \centering
    \includegraphics[width=.5\linewidth]{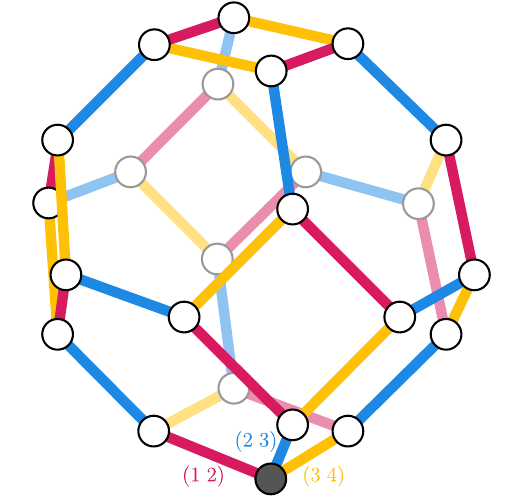}
    \caption{\footnotesize A useful way to visualize a Coxeter system $(W,S)$ is a \emph{Cayley graph}, $(V,E)$, where $V=W$ and
    $(w,w')\in E$ iff there is a generator $s\in S$ such that $w'=ws$. The figure shows the Cayley graph of the 4-letter symmetric group, $A_3$, with generators given by adjacent transpositions. The shaded vertex represents the identity element. The polytope obtained by embedding this graph in $\RR^3$ is called a {\em permutohedron}.
    }
    \label{fig: permutohedron}
\end{figure}

Several well-known structural results about the RM family extend to \emph{any} Coxeter code. First, Coxeter codes are a nested family of codes:
\begin{theorem}\label{thm: nested}
    For integers $q<r\leq m$, the order-$q$ Coxeter code of type $(W,S)$ is strictly contained in the order-$r$ code:
    $$
        \Coxeter{q}\subsetneq\Coxeter{r}.
    $$
\end{theorem}
Like RM codes, Coxeter codes are also closed under duality:
\begin{theorem}\label{thm: Coxeter duality}
    The dual of the order-$r$ Coxeter code of type $(W,S)$ is the corresponding order-$(m-r-1)$ Coxeter code:
    $$
        \Coxeter{r}^\perp = \Coxeter{m-r-1}.
    $$
\end{theorem}

For two vectors $x,y\in \FF^n$, their coordinate-wise (Schur) product is a vector $x\odot y=( x_iy_i, i=1,\dots,n),$ and this definition extends to a product of subsets. 
RM codes satisfy a multiplication property:
for any $r_1,r_2$, 
    $$
    RM(r_1,m)\odot RM(r_2,m)\subseteq RM(r_1+r_2,m)
    $$
 with $RM(r^*,m)\coloneqq \FF^{2^m}$ for all $r^*\geq m$.   This follows since the product of two polynomials of degree $r_1$ and $r_2$ has degree at most $r_1+r_2$. This multiplication property is a general feature of all Coxeter codes:
\begin{theorem}\label{thm: Coxeter multiplication}
    For $r_1, r_2\in \br{-1,\dots,m}$, the Coxeter codes of type $(W,S)$ and orders $r_1$ and $r_2$ satisfy
    $$
        \Coxeter{r_1}\odot\Coxeter{r_2} \subseteq \Coxeter{r_1+r_2},
    $$
    where by convention $\Coxeter{r^*}\coloneqq \FF W$ for all $r^*\geq m$.
\end{theorem}
Lastly, Coxeter codes are (left) ideals in the group algebra, or \emph{group codes} in the sense of Berman \cite{berman1967theory}\footnote{Note a recent paper that extends RM codes \cite{natarajan2023berman}, titled {\em Berman codes}, which is not related to our construction.}. Recall that multiplication in $\FF W$ is given by the convolution of functions, denoted by $f\ast g$.
\begin{theorem}\label{thm: group code}
    For every $f\in\FF W$,  $f\ast\Coxeter{r}\subseteq\Coxeter{r}$. 
\end{theorem}

While \cref{thm: nested,thm: Coxeter multiplication,thm: group code} can be proved using standard tools from group theory and the definition of Coxeter codes, we will delay their proofs as well as the proof of \cref{thm: Coxeter duality} until we have constructed a basis of the codes; 
see \cref{sec: properties}, \cref{prop: AllTheorems} below.

\begin{figure}[t]
    \centering
    \includegraphics[width=.5\linewidth]{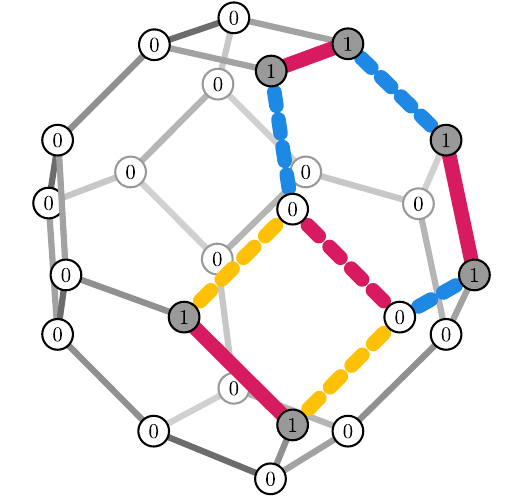}
    \caption{\footnotesize 
    The code $\coxeter{A_3}{S}{1}$ is generated by indicators of faces of the Cayley graph of $A_3$. The bit assignment shown in the figure represents the codeword in $\coxeter{A_3}{S}{1}$ generated by the indicators of the colored hexagonal and square faces. The same codeword is equivalently generated by the indicators of the three solid red edges, indicative of the containment $\coxeter{A_3}{S}{1}\subseteq \coxeter{A_3}{S}{2}$.}
    \label{fig: codes}
\end{figure}

\section{Coxeter Group Preliminaries}
We now list several properties of Coxeter groups in a form and level of generality suitable for our needs.

\begin{lemma}\label{fact: even order}
    A nontrivial finite Coxeter group has even order.
\end{lemma}
\begin{proof}
   If $s\in S\ne\emptyset$, then $\text{ord\,}(s)=2,$ so $\{1,s\}$ is a subgroup of $W$, and the result holds by Lagrange's theorem.
\end{proof}
\begin{lemma}[\hspace*{-.5ex}\cite{BB05}, Prop. 2.4.1]\label{fact: intersection is coset}
    Let $\standard{J_1}$ and $\standard{J_2}$ be standard subgroups, then 
    $$\standard{J_1}\cap \standard{J_2}=\standard{J_1\cap J_2}.$$
    \end{lemma}

\begin{lemma}\label{lem: even overlap}
    Let $\sigma_1\standard{J_1}$ and $\sigma_2\standard{J_2}$ be two standard cosets. If $\abs{J_1}+\abs{J_2}>m$ then $\abs{\sigma_1\standard{J_1}\cap \sigma_2\standard{J_2}}$ is even.
\end{lemma}
\begin{proof}
    The result is true if the cosets have trivial overlap. Otherwise, there is a
$\sigma\in W$ such that 
\begin{align*}
\sigma_1\standard{J_1}\cap \sigma_2\standard{J_2}=\sigma(\standard{J_1}\cap\standard{J_2})
=\sigma\standard{J_1\cap J_2}.
\end{align*}
As $\abs{J_1}+\abs{J_2}>m$ and $\abs{J_1},\abs{J_2}\leq m$, the intersection $J_1\cap J_2$ is non-empty and the result holds by \cref{fact: even order}.
\end{proof}

Coxeter systems carry a natural \emph{length function}, $\ell\colon W\rightarrow \NN$, where the length of an element $w$ is the smallest number of elements from $S$ needed to generate $w$. That is, $\ell(w)=\ell'$ if there is a decomposition $w=\sigma_1\sigma_2\cdots \sigma_{\ell'}$ with $\sigma_i\in S$ for all $i\in[\ell']$, 
and \emph{any} decomposition of $w$ using elements of $S$ contains at least $\ell'$ terms. 
We will make use of two well-known facts:
\begin{lemma}[\hspace*{-.5ex}\cite{BB05}, Lem. 1.4.1]\label{fact: changes length}
    Right multiplication by a generator changes the length of an element, i.e., $\ell(ws)=\ell(w)\pm 1$ for all $w\in W$ and $s\in S$.
\end{lemma}
\begin{lemma}[\hspace*{-.5ex}\cite{AB08}, Prop. 2.20]\label{fact: minimal element}
    A standard coset $w\standard{J}$ has a unique element of minimal length, i.e., there is a unique $w_1\in w\standard{J}$ such that $\ell(w_1)<\ell(u)$ for every $u\in w\standard{J}$. This element is characterized by the property that $\ell(w_1 s)=\ell(w_1)+1$ for every $s\in J$. 
\end{lemma}
Given $w\in W$, these statements suggest a way to construct standard cosets for which $w$ is the minimal element: take $w\standard{J}$ where $J$ is any set of generators that \emph{increase} the length of $w$ via right multiplication. The following standard definition is phrased in terms of elements that \emph{decrease} the length.

\begin{definition}\label{def:descents}
    For $w\in W$, the subset of generators $D(w)\subseteq S$ that reduce the length of $w$ after multiplication on the right is the (right) \emph{descent set} of $w$:\begin{align*}
        D(w)\coloneqq \br{s\in S\bigmid \ell(ws)<\ell(w)}.
    \end{align*}
    The value $d(w)\coloneqq\abs{D(w)}$ is the (right) \emph{descent number} of $w$. 
\end{definition}
\begin{lemma}\label{lem: unique shortest longest}
    For every $w\in W$, $w$ is the unique shortest element of the standard coset $w\standard{S\setminus D(w)}$. 
\end{lemma}
\begin{proof}
    By \cref{fact: changes length}, $\ell(ws)=\ell(w)+1$ for every $s\in S\setminus D(w)$, so the result holds by \cref{fact: minimal element}.
\end{proof}

The following combinatorial quantity will be useful in specifying the dimension of a Coxeter code.
\begin{definition}\label{def:Eulerian}(\cite[Sec.7.2]{BB05},\cite{petersen2015eulerian})
    For $i\in\br{0,\dots,m}$, the $W$-Eulerian number $\euler{W}{i}$ is the count of elements in $W$ with descent number equal to $i$,
    \begin{equation*}
        \euler{W}{i}:=\abs{\br{w\in W\bigmid d(w)=i}}.
    \end{equation*}
\end{definition}    
Eulerian numbers satisfy the {\em Dehn--Sommerville equations}
    \begin{equation}\label{eq: DS}
        \euler{W}{i} = \euler{W}{m-i}.
    \end{equation}
From \cref{def:Eulerian} we also immediately observe that
  \begin{equation}\label{eq: sum}
    \sum_{i=1}^m \euler Wi=|W|.
  \end{equation}
  
\cref{def:descents,def:Eulerian} depend on the choice of generating set $S$, but we suppress this dependence in the notations for simplicity, as is standard. 

\begin{remark}
If $W=\ZZ_2^m$ then $\euler Wi=\binom mi$. If $(W,S)=A_m$ is the symmetric group, then $\euler W i$ is the classic Eulerian number, i.e., the count of permutations in $W$ with $i$ descents \cite[p.6]{petersen2015eulerian}. See \cref{sec: computing Eulerian numbers} for expressions computing $W$-Eulerian numbers for reducible and irreducible Coxeter systems. \hfill$\lhd$
\end{remark}

We conclude this section with a remark on reducible systems. Suppose that 
$(W_1,S_1)$ and $(W_2,S_2)$ are finite Coxeter systems of ranks $m_1$ and $m_2$, respectively. Their direct product $(W,S)\coloneqq (W_1,S)\times (W_2,T)$ is a finite Coxeter system of rank $m_1+m_2$ where $S\coloneqq S\sqcup T$ and $(st)^2=1$ for every $s\in S$ and $t\in T$.  Define the {\em Eulerian polynomial} of the system $W_1$ as
   $$
   W_1(t):=\sum_{i=0}^m \euler {W_1}i t^i,
   $$
and similarly for $W_2$. It is a classic fact \cite[p.202]{BB05} that for the direct product we have 
   \begin{equation}\label{eq: Euler multiplicative}
   W(t)=W_1(t)W_2(t)
   \end{equation}
and thus,
\begin{equation*}
    \euler{W}{k} = \sum_{i+j=k}\euler{W_1}{i}\euler{W_2}{j}, \quad k=1,\dots, s.
\end{equation*}
We will use this property to compute the dimension of codes on products of dihedral groups below.

\section{Code Structure}\label{sec: properties}
In this section, we construct an explicit basis of Coxeter codes, establish their structural properties, and prove the claims stated in \cref{thm: nested,thm: Coxeter duality,thm: Coxeter multiplication,thm: group code}.

\begin{definition}
    For $w\in W$, the \emph{extension} of $w$ in $\FF W$, denoted $\ext{w}\in\FF W$, is the indicator function corresponding to the coset $w\standard{S\setminus D(w)}$, $\ext{w}\coloneqq\indicator{w\standard{S\setminus D(w)}}$. 
    The \emph{rank} of $\ext{w}$ is 
      $$
      \rank(\ext{w}): = m-d(w)=\rank(w\standard{S\setminus D(w)}).
      $$
\end{definition}
\begin{figure}[ht]
\includegraphics[width=.5\linewidth]{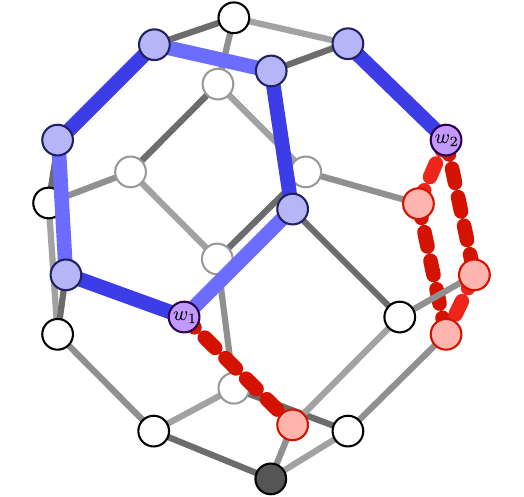}
\caption{This figure shows extensions (blue) and reverse extensions (red) of the elements $w_1$ and $w_2$ in $A_3$.
The identity element is shown as the shaded vertex of the graph.}
\end{figure}
\begin{definition}
Let $\mcB$ denote the set of all extensions. For $i\in\br{0,\dots,m}$, let 
    \begin{align*}
        \mcB_i&:=\br{\ext{w}\in\FF W\bigmid \rank(\ext w)=i}\\
        &=\br{\ext{w}\in\FF W\bigmid w\in W,\;d(w)=m-i}.
    \end{align*}
\end{definition}
Note that by the Dehn--Sommerville equations, \cref{eq: DS}, we have
    \begin{align*}
        \abs{\mcB_{i}}
        = \abs{\mcB_{m-i}} =\euler{W}{i}.
    \end{align*}
For $r\in\br{-1,\dots,m}$ consider the collection of extensions with rank at least $m-r$, 
  $$
  \mcB_{\geq m-r}\coloneqq\bigcup_{i\geq m-r}\mcB_i.
  $$
\begin{example}\label{example: RM4} For the RM case when $W=\ZZ_2^m$, this collection is precisely the standard basis of monomials in $m$ variables with degree at most $r$: if $z\in\ZZ_2^m$ then $\ext{z}=\prod_{i\in\supp(z)}x_i$. For instance, take $m=4$ and let $z=[1001].$ Writing vectors as columns, we have
   \begin{equation*}
    z=\left[\hspace*{-4pt}\begin{array}{c}1\\[-4pt]0\\[-4pt]0\\[-4pt]1\end{array}\hspace*{-4pt}\right],\;S\backslash D(z)=\{e_2,e_3\},\;
    z+\standard{S\backslash D(z)}=\left[\begin{array}{*4{@{\hspace*{1pt}}c}}1&1&1&1\\[-4pt]0&0&1&1\\[-4pt]0&1&0&1\\[-4pt]1&1&1&1 \end{array}\hspace*{-4pt}\right],
    \;\ext z=\indicator{z+\standard{e_1,e_4}}=x_1x_4,
   \end{equation*}
and thus $\mcB_{\ge m-r}$ is equivalently written as the set of monomials of $x_1,\dots,x_4$ of degree $r$ or less. \hfill$\lhd$
\end{example}

We will prove that $\mcB_{\geq m-r}$ is always a basis for the order-$r$ Coxeter code of type $(W,S)$. First, proving that $\mcB$ is linearly independent will rely on the following simple lemma, which says that 
$w\not\in \supp(\ext{u})$ for any $u$ of length at least $w$. Recall again that we do not make a difference between functions and their 
evaluations, so for $u,w\in W$, $\ext{u}(w)=1$ is equivalent to $w\in \supp (\ext u)$.
\begin{lemma}\label{lem: lemma for independence}
    Let $w\in W$ and $U\subseteq W$. If $\ell(w)\leq\ell(u)$ for all $u\in U$ then $\ext{u}(w)=0$ for every $u\in U\setminus\br{w}$.
\end{lemma}
\begin{proof}
    Suppose for contradiction that $\ext{u}(w)=1$ for some $u\in U$, so $w\in u\standard{S\setminus D(u)}$. As $w\neq u$, \cref{lem: unique shortest longest} implies that $\ell(w)>\ell(u)$, contradicting the assumption on $U$.
\end{proof}

\begin{lemma}\label{lem: independence}
    The collection $\mcB$ is linearly independent.
\end{lemma}
\begin{proof}
    Suppose for contradiction that there is a nonempty subset $U\subseteq W$ for which the function $\sum_{u\in U} \ext{u}$ is identically zero. Since $W$ is finite, there must exist a $w\in U$ (not necessarily unique) whose length is minimal among the elements in $U$, i.e., $\ell(w)\leq\ell(u)$ for all $u\in U$. By \cref{lem: lemma for independence} we have $\ext{u}(w)=0$ for all $u\in U\setminus\br{w}$. This, however, is impossible, as it implies $\sum_{u\in U} \ext{u}(w) = \ext{w}(w)=1$.
\end{proof}
We now show that the span of $\mcB_{\geq m-r}$ satisfies a duality structure.
Recall that for two functions $f,g\in\FF W$, their dot product is given by $f\cdot g = \abs{\supp f\cap \supp g}\pmod{2}$.

\begin{lemma}\label{lem: dual spans}
    For each $r\in\br{-1,\dots,m}$ we have
    \begin{equation*}
        \Span \mcB_{\geq m-r} =\left(\Span \mcB_{\geq r+1}\right)^\perp.
    \end{equation*}
\end{lemma}
\begin{proof}
    We first show that $\Span \mcB_{\geq m-r} \subseteq \left(\Span \mcB_{\geq r+1}\right)^\perp$, which is equivalent to the statement that each $\ext{w_1}\in\mcB_{\geq m-r}$ has even overlap with each $\ext{w_2}\in\mcB_{\geq r+1}$. The supports of such $\ext{w_1}$ and $\ext{w_2}$ are standard cosets with ranks $r_1\geq m-r$ and $r_2\geq r+1$, respectively. Since $r_1+r_2>m$, \cref{lem: even overlap} implies that the cardinality of their intersection is even. Thus $\ext{w_1}\cdot\ext{w_2}=0$, as desired.

    We now show $\dim (\Span \mcB_{\geq m-r}) = \dim (\left(\Span \mcB_{\geq r+1}\right)^\perp)$, which implies that the two spaces are, in fact, equal. 
    Using \cref{eq: DS} and the linear independence of $\mcB_{\geq m-r}$, we compute
    \begin{align*}
        \dim (\Span \mcB_{\geq m-r}) &= \sum_{i=m-r}^m \euler{W}{i} 
        =\sum_{i=0}^{r} \euler{W}{i} 
    \end{align*}
    Since the dimensions of a code and its dual code sum to the dimension of the entire vector space, we have
    \begin{align*}
        \dim (\left(\Span \mcB_{\geq r+1}\right)^\perp)) &= \abs{W}-\dim (\Span \mcB_{\geq r+1})
        \\
        &= \sum_{i=0}^{r} \euler{W}{i},
    \end{align*}
where we have used \cref{eq: DS,eq: sum}.
\end{proof}

\begin{theorem}\label{thm: extensions form a basis}
    For $r\in\br{-1,\dots,m}$, $\mcB_{\geq m-r}$ is a basis for the order-$r$ Coxeter code of type $(W,S)$ and rank $m$:
    \begin{equation*}
        \Coxeter{r} = \Span \mcB_{\geq m-r},
    \end{equation*}
    or, alternatively,
    \begin{equation}\label{eq: basis for Coxeter code}
        \Coxeter{r} = \Span \br{\ext{w}\bigmid w\in W,\, d(w)\leq r}.
    \end{equation}
\end{theorem}
\begin{proof}
    Recall that $\Coxeter{r}$ is the span of indicator functions of standard cosets with rank \emph{exactly} equal to $m-r$.
    
($\supseteq$) Consider an $\ext{w}\in \mcB_{\geq m-r}$, which by definition is the indicator function of $w\standard{S\setminus D(w)}$. Let $J\subseteq S\setminus D(w)$ be any subset of $\abs{J}=m-r$ elements of $S\setminus D(w)$, which must exist since $\rank (\ext{w})\geq m-r$. The set of cosets of $\standard{J}$ in $\standard{S\setminus D(w)}$, denoted by $\standard{S\setminus D(w)}/\standard{J}$, forms a partition of $\standard{S\setminus D(w)}$, so
their supports are disjoint, and
    \begin{equation*}
        \ext{w} = \sum_{R\in \standard{S\setminus D(w)}/\standard{J}} \indicator{wR}. 
    \end{equation*}
This shows that $\ext w$ is a sum of standard cosets of rank $m-r$, so it
is a vector in $\Coxeter r$.

    ($\subseteq$) Let $R$ be a standard coset of rank $m-r$ and let $\ext w\in \mcB_{\ge r+1}$. By definition,
    $\rank(\ext w)\ge r+1,$ and thus $\rank(R)+\rank(\ext w)>m$. With this,
\cref{lem: even overlap} implies that $R$ satisfies 
 $\indicator{R}\cdot \ext{w}=0$ for every $\ext{w}\in\mcB_{r+1}$. Thus, $R\in (\Span \mcB_{\ge r+1})^\perp$, which equals $\Span\mcB_{\ge m-r}$ by \cref{lem: dual spans}.
\end{proof}

\begin{example}[\cref{example: RM4} continued]\label{example: RM basis}
If $W=\ZZ_2^m$, then extensions are functions $\ext{z}\colon\ZZ_2^m\rightarrow\FF$ given by $\ext{z}=\prod_{i\in\supp z}x_i$, and descent numbers are given by $d(z)=\abs{z}$. Thus, \cref{eq: basis for Coxeter code} implies that
\begin{equation*}
    \coxeter{\ZZ_2^m}{}{r} = \Span\Big\{\prod_{i\in\supp z}x_i\bigmid z\in\ZZ_2^m,\, \abs{z}\leq r\Big\},
\end{equation*}
proving that $\coxeter{\ZZ_2^m}{}{r}=RM(r,m)$ and recovering the formula $\dim RM(r,m)=\sum_{i=0}^r\binom mi$. \hfill$\lhd$
\end{example}

\begin{proposition}\label{prop: AllTheorems} The following hold for all $q<r$ and $r_1, r_2$:
    \begin{enumerate}
        \item[{\rm(1)}] {\rm(\cref{thm: nested})} $\Coxeter{q}\subsetneq\Coxeter{r}$,
        \item[{\rm(2)}] {\rm(\cref{thm: Coxeter duality})} $\Coxeter{r}^\perp = \Coxeter{m-r-1}$,
        \item[{\rm(3)}] {\rm(\cref{thm: Coxeter multiplication})} $\Coxeter{r_1}\odot\Coxeter{r_2} \subseteq \Coxeter{r_1+r_2}$, and
        \item[{\rm(4)}] {\rm(\cref{thm: group code})} $f\ast \Coxeter{r}\subseteq\Coxeter{r}$ for any $f\in\FF W$.
    \end{enumerate}
\end{proposition}
\begin{proof}
        (1) This follows from \cref{thm: extensions form a basis}.
        
        (2) This follows from \cref{lem: dual spans,thm: extensions form a basis}.
        
        (3) Let $R_1\coloneqq \sigma_1\standard{J_1}$ and $R_2\coloneqq \sigma_2\standard{J_2}$ be standard cosets of ranks $(m-r_1)$ and $(m-r_2)$, respectively, so that $\indicator{R_1}$ and $\indicator{R_2}$ are arbitrary generators of $\Coxeter{r_1}$ and $\Coxeter{r_2}$, respectively. Their intersection, if non-empty, is a standard coset $R_1\cap R_2=\sigma\standard{J_1\cap J_2}$ of rank 
    \begin{align*}
        \abs{J_1\cap J_2} &= |J_1|+|J_2|- \abs{J_1\cup J_2}\\
         &\geq 2m-(r_1+r_2)-m\\
        &=m-(r_1+r_2).
    \end{align*}
    By definition, $\indicator{R_1}\odot\indicator{R_2}=\indicator{R_1\cap R_2}$, and since $R_1\cap R_2$ is a standard coset of rank $\ge m-(r_1+r_2)$, we have $\indicator{R_1\cap R_2}\in\Coxeter{q}$ for some $q\leq r_1+r_2$. The result holds by \cref{thm: nested}.
    
    (4) Suppose that $f=\indicator{w}$ is the indicator function for a single $w\in W$, and that $\indicator{\sigma\standard{J}}$ is the indicator function of an arbitrary rank-$(m-r)$ standard coset. We compute the value of $\indicator{w}\ast\indicator{\sigma\standard{J}}$ on an arbitrary $u\in W$:
    \begin{align*}
        (\indicator{w}\ast\indicator{\sigma\standard{J}})(u) &=\sum_{g\in W}\indicator{w}(g)\indicator{\sigma\standard{J}}(g^{-1}u)\\
        &=\indicator{\sigma\standard{J}}(w^{-1}u)\\
        &=\indicator{(w\sigma)\standard{J}}(u),
    \end{align*}
    where the last line follows since $w^{-1}u\in\sigma\standard{J}$ if and only if $u\in (w\sigma)\standard{J}$. As $(w\sigma)\standard{J}$ is also a rank-$(m-r)$ standard coset, we have that $\indicator{w}\ast\indicator{\sigma\standard{J}}\in\Coxeter{r}$. Since any function can be written in terms of single-point indicators, the full result follows by the linearity of convolution.
 
\end{proof}

\subsection{Reverse extensions}\label{sec: reverse extensions}
We conclude this section with a remark on extensions, which we have chosen to define as indicators corresponding to the cosets $w\standard{S\setminus D(w)}$.
Perhaps a more straightforward choice would have been the cosets corresponding directly to descents, $w\standard{D(w)}$. Indeed, the results of this paper hold equally well by using the \emph{reverse extension}, $\mcR_w\coloneqq\indicator{w\standard{D(w)}}$, e.g.,
\begin{equation*}
    \Coxeter{r} = \Span\br{\mcR_w\bigmid w\in W,\, d(w)\geq m-r}.
\end{equation*}
In the case of RM codes, this basis corresponds to signed monomials of degree at most $r$, $\br{\prod_{i\in A}\bar x_i\mid A\subseteq[m],\, \abs{A}\leq r}$, or equivalently, the evaluation vectors of (unsigned) monomials up to string reversal. Thus, while reverse extensions may appear better suited for the context of Coxeter codes, they do not explicitly generalize the standard basis of RM codes.

\section{Code parameters}
\subsection{Dimension and rate}
\cref{lem: dual spans,thm: extensions form a basis} imply the following result:
\begin{theorem}
 \label{thm: dimension of order r}
    The dimension of the order-$r$ Coxeter code of type $(W,S)$ is given by
    \begin{equation}\label{eq:dimension}
        \dim \Coxeter{r} = \sum_{i=0}^r \euler{W}{i}.
    \end{equation}
\end{theorem}
The rate of the Reed--Muller code $RM(r,m)$ equals $2^{-m}\sum_{k=0}^r\binom mi$. By standard asymptotic arguments,
for large $m$ 
it changes from near zero to near one when $r$ crosses $m/2$, and is about $1/2$ if $r=\lfloor m/2\rfloor$, with more
precise information derived from the standard Gaussian distribution. This behavior largely extends
to many Coxeter codes.

In particular, consider the three infinite series of Coxeter groups in the Coxeter-Dynkin classification:
$A_m$ (the symmetric group on $m+1$ elements), $B_m$ (the hyperoctahedral group of order $2^mm!$), and $D_m$ (the generalized dihedral group of order $2^{m-1}m!$). 
The rate $\kappa(\Coxeter{r})$ has no closed-form expression for any of these cases (for that matter, there is no such expression even for RM codes), but asymptotic normality of Eulerian numbers of types $A,B,D$ has been addressed in many places in the literature \cite{bender1973central}, \cite{chen2009}, with \cite{HCD19} being the most comprehensive source. As implied by these
references, for each of the infinite series of groups, the random variable $X_m$ with $P(X_m=k)=\euler Wk/|W|$ is asymptotically normal with mean $\frac m2$ and variance $\frac m{12}$. Following the proof of the De Moivre--Laplace theorem for the binomial distribution, we obtain the following statement about the asymptotics of the code rate.
\begin{theorem}[{\sc Code rate}]\label{thm: Gaussian} Suppose that $(W,S)_m$ is one of the irreducible Coxeter families $A_m,B_m$, or $D_m$. Let $\Phi(x)=\int_{-\infty}^x e^{-t^2/2}dt/\sqrt{2\pi}$ and let $m\to\infty$.\\
(i) Let $r_m=\frac m2+\rho_m\sqrt{\frac m{12}}$. If $\rho_m\to\rho\in\RR,
$ then the code rate $\kappa(C_W(r_m))\to\Phi(\rho)$. \\
(ii) For a fixed $\kappa\in (0,1)$, define the sequence of order values
  $$
  r_m^\ast:=\Big\lfloor\frac m2+\sqrt{\frac m{12}}\Phi^{-1}(\kappa)\Big\rfloor, m=1,2,\dots.
  $$ 
Assuming that $r_m^\ast\ge 0$,  $\kappa(\C_W(r_m^\ast))\to \kappa$.
\\(iii) Consider a sequence of order values $r_m, m=1,2,\dots$. If $|\frac m2-r_m|\gg \sqrt m$ and for all $m$, 
 (a)
        $r< m/2$, then $\kappa(\C_W(r_m))\to 0$; (b) $r> m/2$, then $\kappa(\C_W(r_m))\to 1$.
\end{theorem}
The rate of any infinite family of Coxeter codes, including the ones constructed from reducible systems (\cref{sec: computing Eulerian numbers}), exhibits a behavior similar to \cref{thm: Gaussian}. This follows from the product structure of the $W$-polynomials of Coxeter groups, \cref{eq: Euler multiplicative}, although the corresponding fact involves convergence to a multivariate Gaussian distribution, as is apparent, for instance, from \cref{eq: dimension dihedral} below.

\subsection{Distance} 
Given that $\Coxeter{r}$ is generated by standard cosets of rank $m-r$, there is a trivial upper bound on the code distance given by the \emph{smallest} such coset. We conjecture that this bound is, in fact, tight:
\begin{conjecture}\label{conj: distance} Let $(W,S)$ be a Coxeter system of rank $m$.
    The distance of the code $\Coxeter{r}$ is given by 
    \begin{equation*}
        \dist(\Coxeter{r})=\min_{J\subseteq S, \abs{J}=m-r} \abs{\standard{J}}.
    \end{equation*}
\end{conjecture}
This conjecture is true for RM codes and the family of Coxeter codes given by the dihedral groups, $I_2(n)$. 
 We have further verified it by computer for all nontrivial Coxeter codes of length at most $120$ (some of them are listed in \cref{tab: Am params,tab: I2(3) params,tab: I2(4) params}, where the distance values shown in italic rely on the validity of \cref{conj: distance}). We can also prove that the conjecture is true whenever $r\geq\lfloor \frac{m}{2}\rfloor$, see \cref{cor: distance for large r} below.

To continue the discussion of the distance, we prove the following lower bound for any $r$:
\vspace*{-.1in}
\begin{theorem}\label{thm: exponential lower bound}
 Let $(W,S)$ be a Coxeter system of rank $m$. The distance of any order-$r$ Coxeter code satisfies $\dist(\Coxeter{r})\geq 2^{m-r}$.
\end{theorem}
This bound is tight for RM codes but not for the codes arising from the symmetric group: the bound in \cref{conj: distance} is strictly larger whenever $r> \lceil \frac{m}{2}\rceil$ in the case of $A_m$.

\begin{lemma}\label{lem: even weight}
    If $r<m$, then for every $c\in\coxeter{W}{}{r}$ the Hamming weight of $c$ is even. 
\end{lemma}
\begin{proof} We know $m-r-1\geq 0$ since $r\leq m-1$, so
    \begin{align*}
        \coxeter{W}{}{r}&\stackrel{\text{(duality)}} =\coxeter{W}{}{m-r-1}^\perp  \\
        &\stackrel{\text{(nesting)}} \subseteq  \coxeter{W}{}{0}^\perp \\
        &\quad= \br{0^{\abs{W}},1^{\abs{W}}}^\perp,
    \end{align*}
    i.e., every $c\in\coxeter{W}{}{r}$ is orthogonal to the all 1's vector and thus has even weight.
\end{proof}

\begin{lemma}\label{lem: disjoint maximal}
    If $w_1,w_2\in W$ are not equal, then there is a $K\subseteq S$, $\abs{K}=m-1$, for which $w_1\standard{K}\neq w_2\standard{K}$.
\end{lemma}
\begin{proof} Let $J_1,\dots,J_{m}$ be the distinct $(m-1)$-subsets of $S$. Note that $\cap_{i=1}^m\standard{J_i}=1$.

Since $w_1\ne w_2$, there is an $i\in[m]$ such that $w_2^{-1}w_1\not\in\standard{J_i}$. Put $K=J_i$ and observe 
that $w_1\standard{K}=w_2\standard{K}$ would yield a contradiction.
\end{proof}

\begin{lemma}\label{lem: induction codeword}
    Consider a standard coset $w\standard{K}$. If $c\in\Coxeter{r}$ then the punctured code $c\vert_{w\standard{K}}\in\coxeter{\standard{K}}{}{r}$.
\end{lemma}
\begin{proof}
    By definition there exist $\br{\sigma_i\standard{J_i}}_{i\in I}$, $\abs{J_i}=m-r$, for which $c=\sum_{i\in I}  \indicator{\sigma_i\standard{J_i}}$. The function restricted to $w\standard{K}$ equals the product $c\indicator{w\standard{K}}$, and
    \begin{align*}
        c\indicator{w\standard{K}} &= \sum_{i\in I} \indicator{\sigma_i\standard{J_i}}\indicator{w\standard{K}},\\
        &= \sum_{i\in I'} \indicator{\sigma_i'\standard{J_i\cap K}},
    \end{align*}
    where $I'\subseteq I$ indexes the standard cosets that have nontrivial intersection with $w\standard{K}$.
    We lower bound
    \begin{align*}
        \abs{J_i\cap K} & = \abs{J_i} + \abs{K} - \abs{J_i\cup K},\\
        &\geq m-r + \abs{K}-m,\\
        &= \abs{K}-r. 
    \end{align*}
Now note that $\coxeter{\standard{K}}{}{r}$ is spanned by standard cosets of rank $\abs{K}-r$. 
By an argument similar to the proof of the first part of \cref{thm: extensions form a basis}, $c\vert_{w\standard{K}}= c\indicator{w\standard{K}}$ is a codeword in $\coxeter{\standard{K}}{}{r}$.
\end{proof}

\begin{proof}[Proof of \cref{thm: exponential lower bound}]
    The result holds for all $m\ge 1$ when $r=0$: $\Coxeter{0}$ is a repetition code with $\dist(\Coxeter{0})=\abs{W}\geq 2^m$. Fix $r\ge 1$. We proceed by induction on $m$. The result is true when $m=r$, as $\Coxeter{m}=\FF W$ has distance $2^0=1$. Supposing that the result holds whenever $(W,S)$ has rank $k\geq r$, consider a system $(W',S')$ with rank $k+1$ and the code $\coxeter{W'}{}{r}$.

  By \cref{lem: even weight}, if $c\in\coxeter{W'}{}{r}$ is a nonzero vector, then $|c|\ge 2$.
  Let $w_1,w_2\in\supp(c)$. By \cref{lem: disjoint maximal} there is a subset $K\subseteq S'$, $\abs{K}=k$
such that $w_1\standard{K}\neq w_2\standard{K}$ (and thus $w_1\standard{K}\cap w_2\standard{K}=\emptyset$). Let $c_1$ and $c_2$ denote the restrictions of $c$ to $w_1\standard{K}$ and $w_2\standard{K}$, respectively. Note the following:
    \begin{enumerate}[leftmargin=*]
        \item By \cref{lem: induction codeword} we are guaranteed that $c_1,c_2\in\coxeter{\standard{K}}{}{r}$.
        \item Since $c(w_1)=c(w_2)=1$, these restrictions are nonzero codewords of $\coxeter{\standard{K}}{}{r}$.
        \item Since $w_1\standard{K}\neq w_2\standard{K}$, their intersection is empty, and we obtain $\abs{c} \geq \abs{c_1}+\abs{c_2}$.
    \end{enumerate}
    Since the rank of $(\standard{K},K)$ is $k$, we can use the induction hypothesis for $c_1$ and $c_2$, which are nonzero codewords of $\coxeter{\standard{K}}{}{r}$, to obtain
    \begin{align*}
        \abs{c} \geq \abs{c_1}+\abs{c_2} \geq 2^{k-r} + 2^{k-r} = 2^{k+1-r},
    \end{align*}
    completing the proof.
\end{proof}

\begin{corollary}\label{cor: distance for large r}
    If $r\geq\lfloor \frac{m}{2}\rfloor$ then $\dist(\Coxeter{r})=2^{m-r}$.
\end{corollary}
\begin{proof} Let $\Coxeter r$ be a code of order $r$ constructed from a Coxeter system $(W,S)$.
If there is a standard subgroup $\standard J$
of rank $m-r$, all of whose generators are pairwise commuting, this yields a codeword of weight $2^{m-r}$, matching the lower bound from \cref{thm: exponential lower bound}. By assumption, $m-r\le
\lceil\frac m2\rceil$, so our claim will follow if we show that any Coxeter system contains at least $\lceil \frac m2\rceil$
commuting generators.

First, suppose that $(W,S)$ is irreducible. As mentioned above, irreducible systems are
completely classified in terms of their Coxeter-Dynkin diagrams \cite{BB05}. Any such diagram is connected and, by inspection, has no cycles. In other words, it is a bipartite graph, which therefore contains a part of size $\ge \lceil\frac m2\rceil$. This subset of vertices forms an independent set, giving the desired collection of commuting generators.

Now suppose that $(W,S)=\prod_i (W_i,S_i)$, where each factor is irreducible, and let $m_i:=|S_i|$ for all $i$, so that $|S|=\sum_i m_i$. Generators from different sets $S_i$ commute, and each $S_i$ contains $\ge \lceil \frac {m_i}2\rceil$ commuting generators by the above. Since 
    $$
    \sum_i \Big\lceil \frac {m_i}2\Big\rceil\ge \biggl\lceil\frac{\sum_i m_i}2\biggr \rceil,
    $$
this again proves our claim.
\end{proof}

Supposing that \cref{conj: distance} is true, we will compute the distances for two particular families--- those of type $A_m$ and $I_2(n)^\mu$.

\subsubsection{Codes of type $A_m$.} 
For $m\ge 1$, $A_m$ is a rank-$m$ Coxeter system with defining matrix
\begin{equation*}
    M(i,j) = \begin{cases}
        1, &i=j,\\
        3, &\abs{j-i}=1,\\
        2, &\text{otherwise}.
    \end{cases}
\end{equation*}

For $r\in\{1,\dots,m\}$ let 
   $$
  T(m,r):= \Big(\Big\lceil\frac{m}{r}\Big\rceil!\Big)^{m\,\text{mod }r}
     \Big(\Big\lfloor\frac{m}{r}\Big\rfloor!\Big)^{r-m\,\text{mod }r}.
     $$
Note that 
   $$
  \Big \lceil\frac{m}{r}\Big\rceil (m\,\text{mod\,}r )
  +\Big\lfloor\frac{m}{r}\Big\rfloor( r-m\,\text{mod\,}r)=m
  $$
and that this relation describes a partition of $m$ into $r$ close-to-equal parts with the largest possible number of 
parts of size $\lfloor\frac{m}{r}\rfloor$.

\begin{theorem} \label{thm: symmetric order 1 parameters}
The parameters {\rm[length, dimension, distance]} of the codes $\coxeter{A_m}{}{r}$ for all
$r\in\br{0,1,\dots,m}$ are given by:
 $$
 \bigg[(m+1)!,\; \sum_{i=0}^r \euler{A_m}{i},\; T(m+1,r+1)\bigg]
 $$
assuming Conjecture~\ref{conj: distance} when $r<\lfloor\frac{m}2\rfloor$.
\end{theorem}
\begin{proof}
The length and dimension are immediate from the construction. To find the code distance, first let $r\ge \lfloor\frac m2\rfloor$. In this case, \cref{cor: distance for large r} implies that $\dist_{A_m}(r)=2^{m-r}$. We will show
that $T(m+1,r+1)=2^{m-r}$. To see this, we consider the following two possibilities:

\noindent (a) If
$r\ge \lfloor\frac m2\rfloor+1$, then $\lceil \frac {m+1}{r+1}\rceil=2$, $\lfloor\frac{m+1}{r+1}\rfloor=1$, and
$(m+1)\,\text{mod}(r+1)=m-r$. 

\noindent (b) If $r=\lfloor\frac m2\rfloor$, then 
\begin{itemize}
 \item[(b1)] if $m$ is odd, then $\lceil\frac{m+1}{r+1}\rceil=\lfloor\frac{m+1}{r+1}\rfloor=2$, and their exponents in the expression for
 $T(m+1,r+1)$ are 0 and $m-r$,
respectively; 
 \item[(b2)] if $m$ is even, then $\lceil\frac{m+1}{r+1}\rceil=2$, $(m+1)\,\text{mod}(r+1)=m-r$, and $\lfloor\frac{m+1}{r+1}\rfloor=1$, confirming again the value of $2^{m-r}$.
 \end{itemize}
Altogether, this shows our claim.

Now let $r\le \lfloor\frac m2\rfloor-1$ or $m-r\ge \lfloor \frac {m+1}2\rfloor+1$. In this case, some of the generators of any rank-$(m-r)$ subgroup necessarily do not commute since the transpositions overlap. Suppose that disjoint sets $S_1,S_2,\dots,S_{r+1}$ 
form a partition of $[m+1]$ into $r+1$ segments, wherein the junction points of the segments correspond to the $r$ missing generators in the set of $m-r$ generators. Each set $S_i$ generates a permutation group of order $|S_i|!$, and the order of $H$ equals the product of their orders. This product is minimized if its terms are equal, or as close as possible to being equal, i.e., $S_i\in\{\lfloor \frac{m+1}{r+1}\rfloor, \lceil\frac{m+1}{r+1}\rceil\}$ with as many
smaller-size subsets $S_i$ as possible. According to the remark before the theorem, the size of $H$ is exactly $T(m+1,r+1)$, 
and Conjecture~\ref{conj: distance} implies that this is the value of the code distance.
\end{proof}

Note that in the $r\ge\lfloor\frac ms\rfloor$ case of this theorem, the subgroup $H$ is generated by commuting transpositions and therefore forms an $(m-r)$-dimensional cube in the Cayley graph, giving rise to a minimum-weight codeword in $C_{A_m}(r)$. In the Reed-Muller case, since all the generators commute, the distance of the code is exactly $2^{m-r}$ for all $r$.

\begin{remark} The sequence $T(1,1),T(2,1),T(2,2), T(3,1),T(3,2),\dots$ appears  in OEIS \cite{oeis} as entry A335109. According to the OEIS description, the number $T(m,r)$ gives the count of permutations $\pi:[m]\to[m]$ such that $\pi(i)\equiv i(\text{mod\,}r)$ for all $i\in[m]$. It is not clear to us if the two descriptions are connected.

\end{remark}

The code $\coxeter{A_m}{}{0}$ of order $r=0$ is simply a repetition code. The parameters of the first-order code
can be written explicitly as follows.
\begin{proposition}\label{thm: FirstOrderCodes}
For $m\ge 1$, the parameters of the binary linear code $\coxeter{A_m}{}{1}$ (assuming \cref{conj: distance}) are given by:
    $$ 
    \Big[(m+1)!, 2^{m+1}-m-1, {(m+1)!}\Big/\binom{m+1}{\lfloor \frac{m+1}{2}\rfloor}\Big].
    $$
\end{proposition}
\begin{proof}
The dimension $\dim(\coxeter{A_m}{}{1})=1+\euler{A_m}{1}$. The Eulerian number $\euler{A_m}{1}$ can be found using \cref{eq: Euler A} below:
    $$
    \euler{A_m}{1}=\sum_{i=0}^{m-1} (m-i)2^i=2^{m+1}-m-2,
    $$
giving the value of the dimension.
 The sequence of distances 
    $\dist(\coxeter{A_m}{}{1})= T(m+1,2)$ appears as entry A010551 in OEIS \cite{oeis}, and has explicit formula 
    $T(m,2) = m!/\binom{m}{\lfloor \frac{m}{2}\rfloor}.$
\end{proof}

\subsubsection{Codes of type $I_2(n)^\mu$.} For $n\in\ZZ_{\geq 2}$ and $m\ge 1$, $I_2(n)^\mu$ is a Coxeter system of rank $m=2\mu$ with $\abs{I_2(n)^\mu}=(2n)^\mu$ and defining matrix
\begin{equation*}
    M(i,j) = \begin{cases}
        1, &i=j,\\
        n, &j=i+1 \text{ and } j\equiv 0\pmod{2},\\
        n, &i=j+1 \text{ and } i\equiv 0\pmod{2},\\
        2, &\text{otherwise}.
    \end{cases}
\end{equation*}

\begin{proposition}\label{prop: I2(n)}
The binary linear code $\coxeter{I_2(n)^\mu}{}{r}$ has parameters $[(2n)^\mu,k,d]$, where the dimension $k$ is given by
\begin{align}
    k &= \sum_{\substack{i,j\in\NN\\ i+j\leq \mu\\2i+j\leq r}} \frac{\mu!}{i!j!(\mu-i-j)!}(2n-2)^j  \label{eq: dimension dihedral}\\
    \intertext{and the distance $d$ (assuming \cref{conj: distance} when $r<\mu$) is given by}
    d &= \begin{cases}
        2^{2\mu-r}, &\mu\le r \le 2 \mu,\\
        2^\mu n^{\mu-r}, &0\le r< \mu.
    \end{cases} \notag
\end{align}
\end{proposition}
\begin{proof} 
For the dimension, we note that the Eulerian numbers of $I_2(n)$ are $\euler Wi=1, 2n-2, 1$ for $i=0,1,2$, so using 
\cref{eq: Euler multiplicative}, we obtain
  $
  W(t)=(t^2+(2n-2)t+1)^\mu
  $
Computing the dimension of the code $\coxeter{I_2(n)^\mu}{}{r}$ by \cref{eq:dimension}, we obtain the expression
in \cref{eq: dimension dihedral}.
  
 Turning to the distance, the $r\geq \mu$ case holds by \cref{cor: distance for large r} (note that the rank of this Coxeter system is $2\mu$), so we only rely on \cref{conj: distance} when $r<\mu$. We need to minimize the size of $\abs{\standard{J}}$ where $J\subset S, \abs{J}=2\mu-r\geq \mu$. It is straightforward to verify that, without loss of generality, such a collection necessarily contains the even index generators, $J_{\mathrm{even}}=\br{2i}_{i=1}^\mu\subseteq J$. For each additional generator $s_{2j-1}$ added to $J_{\mathrm{even}}$, we replace a factor of $2$ in $\abs{\standard{J}}$ with a factor of $2n$, the order of the subgroup $\standard{s_{2j-1},s_{2j}}$.
\end{proof}
\begin{corollary}
For fixed $r,n$ and $m\to\infty$, the distance of $\coxeter{I_2(n)^\mu}{}{r}$ is $(2n)^m n^{-r}$, i.e., it forms a 
constant proportion of the code length.
\end{corollary}
Codes $\coxeter{I_2(n)^\mu}{}{r}$ are perhaps the closest to RM codes in the Coxeter family: for instance, 
$\coxeter{I_2(2)^\mu}{}{r}$ is simply $RM(r,\mu)$, so it is of interest to further study such codes for small $n$.
In \cref{sec: I23} we give a table of parameters of the codes $\coxeter{I_2(n)^\mu}{}{r}$ for $n=3,4$ and several values of $\mu$.

\section{Computing \texorpdfstring{$W$}{W}-Eulerian Numbers}\label{sec: computing Eulerian numbers}

To find the code dimension via \cref{eq:dimension}, it is useful to have explicit expressions for the $W$-Eulerian numbers. For the irreducible families of Coxeter groups, they appear in many references, e.g., \cite{petersen2015eulerian,hyatt2016recurrences,BRENTI1994417}. We give these expressions in our notation, along with an expression to compute the $W$-Eulerian numbers for direct products of Coxeter groups.

For every finite Coxeter system $(W,S)$ of rank $m$, the 0-th and $m$-th $W$-Eulerian numbers equal 1, $\euler{W}{0}=\euler{W}{m}=1$.

\vspace{0.5em}
\noindent {\bfit Type A.} \cite[A008292]{oeis} The $A_m$-\emph{Eulerian numbers} can be computed via the recurrence relation
\begin{equation}\label{eq: Euler A}
    \euler{A_m}{i} = (m-i+1)\euler{A_{m-1}}{i-1} + (i+1)\euler{A_{m-1}}{i}.
\end{equation}

\vspace{0.5em}
\noindent {\bfit Type B.} \cite[A060187]{oeis} The $B_m$-Eulerian numbers can be computed via the recurrence relation
\begin{equation*}
    \euler{B_m}{i} = (2m-2i+1)\euler{B_{m-1}}{i-1} + (2i+1)\euler{B_{m-1}}{i}.
\end{equation*}

\vspace{0.5em}
\noindent {\bfit Type D.} \cite[A066094]{oeis} The $D_m$-Eulerian numbers can be computed from the $A_m$- and $B_m$-Eulerian numbers via
\begin{equation*}
    \euler{D_m}{i} = \euler{B_m}{i}-m 2^{m-1}\euler{A_{m-2}}{i-1}.
\end{equation*}

\vspace{0.5em}
\noindent {\bfit Dihedral group.}  Since $I_2(n)$ has two generators, the only possible descent numbers are 0, 1, and 2, so $\euler{I_2(n)}{1} = 2n-2.$

\vspace{0.5em}
\noindent {\bfit Exceptional types.} See \cref{tab: exceptional}.

{\begin{table*}[h!]
    \centering
{ \begin{tabular}{L{2em} c c c c c c c }
    \hline
   & & & & $r$& & &    \\ \cline{2-8}
 $W$  & 1& 2& 3& 4& 5& 6& 7   \\\hline\hline
        $E_6$  & 1272 &12183 & 24928&12183 &1272 &1 &  \\ \hline
        $E_7$  & 17635 & 309969 &1123915  & 1123915 & 309969 & 17635&1   \\ \hline
        $E_8$  & 881752& 28336348& 169022824&300247750  & 169022824& 28336348& 881752  \\ \hline
        $F_4$  & 236& 678 & 236&1 & & &        \\ \hline
        $H_3$  & 59& 59& 1& & & &     \\ \hline
        $H_4$  & 2636& 9126& 2636& 1& & &    \\ \hline
    \end{tabular}
    \caption{$W$-Eulerian numbers for groups of exceptional type \cite[p.248]{petersen2015eulerian}.}
    \label{tab: exceptional}
    }
\end{table*}
}

\section{Quantum codes from Coxeter groups}\label{sec:quantum}
We adopt conventions from \cite{Gottesman2024}.
Denote by $[[n,k]]$ the parameters of a qubit stabilizer code that encodes $k$ logical qubits into $n$ physical qubits. Given binary $[n,k_i]$ codes $C_i$, $i\in\br{1,2}$, such that $C_1^\perp\subseteq C_2$ there is an $[[n,k_1+k_2-n]]$ stabilizer code, known as the CSS code associated to $C_1$, $C_2$, denoted by $\CSS(C_1,C_2)$. The codes $C_1^\perp$ and $C_2^\perp$ represent the $X$ and $Z$ stabilizers of $\CSS(C_1,C_2)$, respectively. That is, denoting $X^x\coloneqq \bigotimes_{i\in[n]}X^{x_i}$ and $Z^z\coloneqq \bigotimes_{i\in[n]}Z^{z_i}$ where $X$ and $Z$ are the Pauli matrices, the operators
  \begin{equation}\label{eq: XZ}
    \br{X^x, Z^z\Bigmid x\in C_1^\perp, z\in C_2^\perp},
 \end{equation}
commute and have a joint $+1$ eigenspace in $\CC^{2^n}$ of dimension $2^{k_1+k_2-n}$. The codes $C_1$ and $C_2$ likewise represent the space of logical $Z$ and $X$ Pauli operators, respectively.

Let $(W,S)$ be a finite Coxeter system of rank $m\geq 1$.
For $-1\leq q\leq r\leq m$, \cref{thm: nested} implies that $\Coxeter{q}\subseteq\Coxeter{r}$, and so we immediately construct a quantum code using Coxeter codes:
\begin{definition}[Quantum Coxeter code]\label{def: Quantum Coxeter}
     The \emph{order-$(q,r)$ quantum Coxeter code of type $(W,S)$}, $\QCoxeter{q,r}$, is defined to be the CSS code $$\QCoxeter{q,r}\coloneqq\CSS(\Coxeter{m-q-1},\Coxeter{r})$$ with parameters $[[n=\abs{W}, k=\sum_{i=q+1}^r\euler{W}{i}]]$.
\end{definition}
\begin{theorem}\label{thm: QCoxeter parameters} The parameters of the $\QCoxeter{q,r}$ are 
     $$
     [[n=\abs{W}, k=\sum_{i=q+1}^r\euler{W}{i},d= 2^{\min(q+1,m-r)}]].
     $$
\end{theorem}
\begin{proof}
    The length and dimension are clear by construction. Using the notation introduced in the beginning of this
    section,
   $
C_{1}=\Coxeter{m-q-1}, \quad C_{2}=\Coxeter{r}.
   $
The distance 
    $\dist(\QCoxeter{q,r})=\min(d_X,d_Z)$, where $d_X:=w_H(C_1\backslash C_2^\bot)$ is the minimum Hamming weight of the binary code $C_1\backslash C_2^\bot$ and similarly for $d_Z:=w_H(C_2\backslash C_1^\bot)$.
Below we assume that $q<r$ because if $q=r$, then the dimension of the code $k=0$, and the distance is not well defined. The argument depends on whether $r\le\lfloor\frac m2\rfloor$ or not.

1. $q< r\le\lfloor\frac m2\rfloor$. In this case, $m-q-1\ge m-\lfloor\frac m2\rfloor\ge \lfloor\frac m2\rfloor$, and thus
$\dist(C_1)= 2^{q+1}$ by \cref{cor: distance for large r}, and $\dist(C_2^\bot)= 2^{m-r}$ for the same reason. Since $C_2^\bot\subseteq C_1,$ we conclude that $d_X= 2^{\min(q+1,m-r)}$.
The argument for $d_Z$ is fully analogous, which proves the claim of the theorem.

2. $q\le \lfloor\frac m2\rfloor< r$. As above, we have $\dist(C_1)=2^{q+1}$. By \cref{thm: exponential lower bound}, $\dist(C_2^\bot)\ge 2^{r+1}\ge\dist(C_1),$ so clearly $d_X=2^{q+1}$. The argument
for $d_Z$ is again fully symmetric, yielding the estimate $d_Z=2^{m-r}$ and concluding the proof.    
\end{proof}

Consider $n=\abs{W}$ physical qubits indexed by the elements of $W$. For a subset $A\subseteq W$ let $X_{A}$ denote the $n$-qubit Pauli operator acting as $X$ on the qubits in $A$ and $\eye$ (identity) elsewhere, and analogously for $Z_{A}$.
The next lemma is a simple consequence of the definition of classical Coxeter codes and their duality structure given in \cref{thm: Coxeter duality}.
\begin{lemma}\label{lem: Quantum Coxeter}
    Given $q,r\in\br{-1,\dots,m}$, $q\leq r$, the following collections of $X$ and $Z$ operators generate the stabilizers of $\QCoxeter{q,r}$:
    \begin{equation*}
        \begin{aligned}
        \mcS_X &\coloneqq\Big\{X_{w\standard{J}} \bigmid w\in W,J\subseteq S, \abs{J} = m-q\Big\},\\
        \mcS_Z &\coloneqq\Big\{Z_{w\standard{J}} \hspace{0.15em}\bigmid w\in W,J\subseteq S, \abs{J} = r+1\hspace{0.4em}\Big\}.
        \end{aligned}
    \end{equation*}
\end{lemma}
As a simple example, consider the dihedral group $I_2(n)$ whose Cayley graph is a $2n$-cycle. Then $\QCoxeter{0,1}$ is the Iceberg code generated by global $X^{\otimes 2n}$ and $Z^{\otimes 2n}$ stabilizers.

In prior work \cite{barg2024geometric}, we utilized the geometric and combinatorial structure of the group $\ZZ_2^m$ with its standard generating set to study transversal logical operators in higher levels of the Clifford hierarchy of the quantum RM family, $QRM_m(q,r) = \Qcoxeter{\ZZ_2^m}{}{q,r}$. For instance, the exact nature of the logic implemented by certain transversal operators acting on a standard coset depends only on the rank of the coset. This result holds in the case of arbitrary quantum Coxeter codes.
\begin{claim}\label{lem: validity} Let $\QCoxeter{q,r},0\le q<r\le m$ be the quantum Coxeter code and let $R$ be a standard coset.
    For the single-qubit operator 
    $$Z(k)\coloneqq\ketbra{0}+e^{i\frac{\pi}{2^k}}\ketbra{1},$$
    \begin{enumerate}[leftmargin=*]
        \item If $\rank(R) \leq q+kr$, then applying $Z(k)$ to the qubits in $R$ does not preserve the code space.
        \item If $q+kr+1\leq \rank(R)\leq (k+1)r$, then applying $Z(k)$ to the qubits in $R$ implements a non-trivial logical operation the code space.
        \item If $\rank(R)\geq (k+1)r+1$, then applying $Z(k)$ to the qubits in $R$  implements a logical identity on the code space.
    \end{enumerate}
\end{claim}
The proof of \cref{lem: validity} is identical to the proof of Theorem 5.2 in \cite{barg2024geometric}, which relies only on the Coxeter group structure of $\ZZ_2^m$. A natural future direction, following the main results of \cite{barg2024geometric}, is to give a combinatorial description of the logical circuit implemented by a $Z(k)_R$ operator when $q+kr+1\leq \rank(R)\leq (k+1)r$. A necessary first step would be to construct a so-called ``symplectic basis'' for $\QCoxeter{q,r}$, i.e., a set of Pauli operators that generate the space of logical Paulis and satisfy certain commutativity conditions. In a few cases, the collections of forward and reverse extensions satisfy the symplectic condition.

At the same time, in many cases, this fails to be true, including some small quantum Coxeter codes. Examples of groups for which the symplectic condition fails include the system $(A_3,S)$ considered above (the symmetric group on 4 letters), and $B_2$, the dihedral group of order 8 generated by two reflections across lines in $\RR^2$ that meet at a 45{\textdegree}  angle. 

The codes $\QCoxeter{0,1}$ for the Coxeter systems $A_3$, $B_3$, and $H_3$ appear in \cite{Vasmer2022} as examples of 3D ball codes. The authors of \cite{Vasmer2022} note that a global transversal $T$ operator is a non-trivial logical operator for these codes; this is also a consequence of our \cref{lem: validity}.\footnote{\cite{Vasmer2022} technically considers a \emph{signed} version of transversal $T$, which acts as $T$ on half of the qubits and $T^\dagger$ on the remaining qubits. Our \cref{lem: validity} applies in this case, as well.}

\begin{remark}
A related construction of quantum stabilizer codes was earlier outlined in \cite{vuillot2022quantum}. Its authors start with an abstract combinatorial generalization of RM codes wherein the group $\ZZ_2^m$ is replaced with a Cartesian product 
$\mathcal{L}_m=L_1\times\dots \times L_m$ of  finite sets of varying size. Fixing a subset $\mathcal{F}\subset \mathcal{L}_m$ defines the support set of qubits of the quantum code, and the stabilizers act on specially chosen subsets of $\mathcal{F}$ that sustain the commutation relations. As the authors of \cite{vuillot2022quantum} observe, one way of choosing the collection $\mathcal{L}_m$ is by taking the sets $L_i$ as rank-$(m-1)$ standard subgroups of a Coxeter group $W$ of rank $m$. They further construct the stabilizer group by taking $X$- and $Z$-stabilizers that act on subsets corresponding to the standard cosets of $W$. At the same time, \cite{vuillot2022quantum} does not link this construction to CSS codes or identify the properties of the obtained quantum codes, suggesting that knowing the group presentation is not sufficient for that purpose. Our approach advances this understanding, showing that it is possible to pinpoint code's properties starting from the structure of the underlying Coxeter group.
\end{remark}

\subsection{The dihedral (quantum) code family}
Examples of quantum codes $\QCoxeter{q,r}$ can be obtained using parameters of classical codes listed below in Tables~\ref{tab: Am params}--\ref{tab: I2(4) params} relying on \cref{thm: QCoxeter parameters}. Here, we focus on the case $W=I_2(n)^\mu$: $\mu$ copies of the $2n$-element dihedral group for $\mu\ge 2$. The quantum code $\Qcoxeter{W}{}{q,r}$ is obtained as $\text{CSS}(\coxeter{W}{}{2\mu-q-1),\coxeter{W}{}{r}}$, so to find its parameters explicitly, we rely on the parameters $[(2n)^\mu,k,d]$ of
classical dihedral Coxeter codes $\coxeter{W}{}{\cdot}$ as given in \cref{prop: I2(n)}.  
For a concrete example, consider the case $r=\mu, q=\mu-1, n=3$. Then the parameters of the code $\mcQ_\mu:=\Qcoxeter{I_2(3)^\mu}{}{q,r}$ are
    $$
   \Big[\Big[\text{length}=6^\mu,k=\euler{I_2(3)^\mu}{\mu},d=2^\mu\Big]\Big].
   $$
The dimension $k$ can be computed explicitly: recalling the proof of \cref{thm: QCoxeter parameters}, this is simply
the ``central coefficient'' in the expansion of the Eulerian polynomial $W(t)$:
     \begin{align*}
      \dim(\mcQ_\mu)&=\text{Coeff}_{[t^\mu]}(t^2+4t+1)^\mu=\sum_{i,j,l}\frac{\mu!}{i!j!l!}4^j,
\end{align*}
where $i,j,l\ge 0$ and $i+j+l=\mu, 2i+j=\mu$. Solving for $j, l$, we obtain
$l=i,j=\mu-2i$. Substitute into the above line and rewrite to obtain the expression 
  \begin{equation}\label{eq: dim Q}
   \dim(\mcQ_\mu)=\sum_{i=0}^{\lfloor \mu/2\rfloor} \frac{\mu!}{(i!)^2(\mu-2i)!}4^{\mu-2i}.
   \end{equation}

Let us compare the obtained parameters with existing proposals. A family of codes with similar parameters was considered recently in \cite{goto2024high}. The codes in this family, which the authors refer to as \emph{many-hypercube codes}, are obtained as concatenations of $\mu$ copies of the $[[6,4,2]]$ Iceberg code, i.e., concatenations of $\Qcoxeter{I_2(3)}{}{0,1}$, resulting in parameters $[[6^\mu,4^\mu,2^\mu]]$ for all $\mu\ge 2$. 

Clearly, the codes $\mcQ_\mu$ have the same length and distance as the many-hypercube codes. Isolating the first two
terms in \cref{eq: dim Q}, we further obtain
    $$
    \dim(\mcQ_\mu)\ge \Big(1+\frac{\mu(\mu-1)}{16}\Big)4^\mu,
    $$
where the inequality is strict for all $\mu\ge 4$. For the same values of length and distance, quantum (dihedral) Coxeter codes $\mcQ_\mu$ encode strictly more logical information than the construction of \cite{goto2024high} for all $\mu>1$.

One may wonder how the information rates of these two code families compare as $\mu$ increases. For the many-hypercube codes, the rate declines exponentially as $(2/3)^\mu$. To compute the rate asymptotics of the $\mcQ_\mu$ family, we have to analyze the behavior of the sum in \cref{eq: dim Q}, relying on the generating function of the ``central trinomial coefficients'' \cite{wagner2012asymptotics}. As a result, we obtain $\Theta(\mu^{-1/2})$, so the rate of quantum Coxeter codes, while not constant, exhibits a much slower decline.

Let us give a few numerical examples using \cref{tab: I2(3) params}. It is easier to find the code dimension once we realize that $k=\dim(\coxeter{I_2(3)^\mu}{}{\mu})-\dim(\coxeter{I_2(3)^\mu}{}{\mu-1})$. For instance, for $\mu=3,4$, the codes $\mcQ_\mu$ have parameters $[[216,88,8]]$ and $[[1296,454,16]]$. At the same time, the many-hypercube codes for the same $\mu$ have parameters $[[216,64,8]]$ and $[[1296, 256, 16]]$.

Note that the distance of the code $\mcQ_3=\Qcoxeter{I_2(3)^3}{}{2,3}$ still falls short of the best known quantum code \footnote{per codetables.de; the code was constructed by computer. The tables stop at length $n=256$.} for $n=216,k=88$, which has distance 21. At the same time, both Coxeter and many-hypercube codes are instances of general code families with clearly described structure, and in the latter case are also equipped with efficient encoding and decoding procedures.

\vspace*{-.1in}
\section{Code examples}
One particularly useful way to visualize Coxeter groups and codes is through the notion of a Cayley graph.
\begin{definition}
    The \textit{Cayley graph} of a Coxeter system $(W,S)$ is a graph $G=(V,E)$ with vertices given by elements of the group $V\coloneqq W$, and with edges given by $$E\coloneqq\br{(w,v)\mid w^{-1}v\in S}.$$
\end{definition}

\begin{figure}[!ht]
    \centering
    \includegraphics[width=\linewidth]{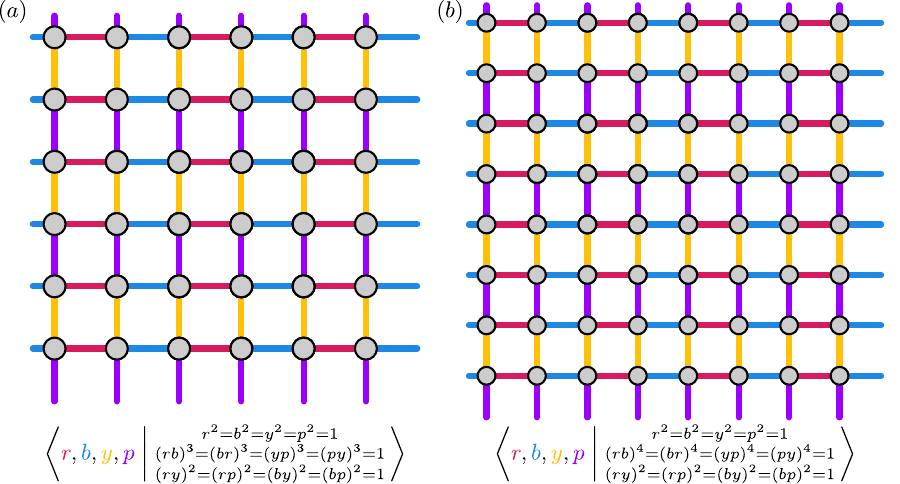}
    \caption{\footnotesize Cayley graphs for Cartesian products of two dihedral groups: (a) $I_2(3)$--- note that $I_2(3)\cong A_2$, the symmetric group on 3 letters--- and (b) $I_2(4)$. The Coxeter system $I_2(4)\cong B_2$, the hyperoctahedral group, or \emph{signed symmetric group}, on 3 letters.}
    \label{fig: I2 figs}
\end{figure}

\begin{figure}[!ht]
    \centering
    \includegraphics[width=\linewidth]{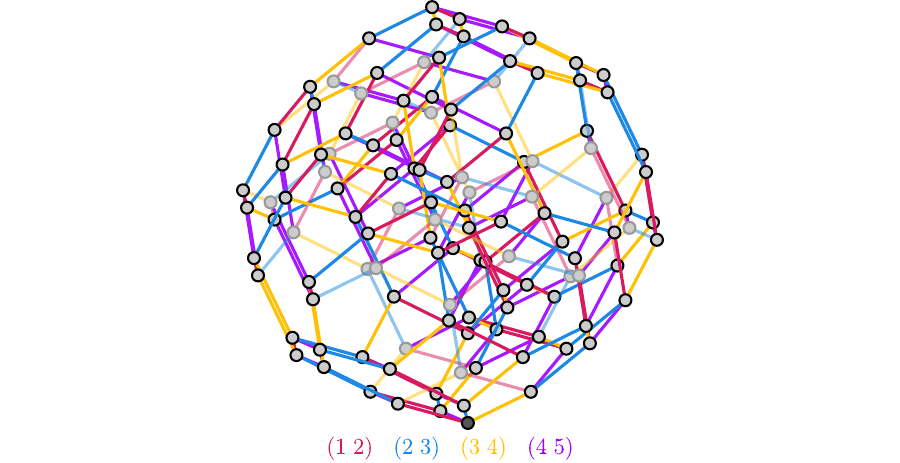}
    \caption{\footnotesize Cayley graph for the symmetric groups $A_4$}
    \label{fig: Am figs}
\end{figure}

The Cayley graph of a Coxeter group is undirected since each generator squares to identity, and it also has a natural edge-coloring given by $\mathrm{color}((w,v))\coloneqq w^{-1}v$. 

Below we consider some Coxeter codes arising from the families $A_m$, $I_2(3)^\mu$, and $I_2(4)^\mu$. In addition to showing Cayley graphs for some of these groups, we also list some explicit code parameters. Italics indicate distances that rely on \cref{conj: distance} and regular font indicates a proven value. In particular, \cref{cor: distance for large r} guarantees that $\dist(\Coxeter{r})=2^{m-r}$ whenever 
$r\geq\lfloor\frac m2\rfloor$; the distances of some order-1 codes were computed by brute force.

\subsection{Codes of type \texorpdfstring{$A_m$}{Am}}

Consider Coxeter codes corresponding to the infinite family $A_m$, the symmetric group on $m+1$ letters. The Cayley graphs for $A_3$ and $A_4$ are shown in \cref{fig: permutohedron,fig: Am figs}, respectively.

\begin{table}[ht!]
    \centering
  
    \begin{tabular}{L{0.05\textwidth} c c c c c}
    \hline

    & $m$ &  &  &  &  \\ \cline{2-6}
    $r$ & 2 & 3 & 4 & 5 & 6 \\ \hline\hline
    1 & $[6 ,\,5 ,\,{2}]$ & $[24 ,\,13 ,\,{4} ]$ & $[120 ,\,27 ,\,{12} ]$ & $[720 ,\,58 ,\,\emph{ 36} ]$ & $[5040 ,\,121 ,\,\emph{ 144} ]$ \\ \hline
    2 & $[6 ,\,6 ,\,{1} ]$ & $[24 ,\,23 ,\,{2} ]$ & $[120 ,\,93 ,\,{4} ]$ & $[720 ,\,360 ,\,{8} ]$ & $[5040 ,\,1312 ,\,\emph{ 24} ]$ \\\hline
    3 & & $[24 ,\,24 ,\,{1} ]$ & $[120 ,\,119 ,\,{2} ]$ & $[720 ,\,662 ,\,{4} ]$ & $[5040 ,\,3728 ,\,{8} ]$ \\\hline
    4 & & & $[120 ,\,120 ,\,{1} ]$ & $[720 ,\,719 ,\,{2} ]$ & $[5040 ,\,4919 ,\,{4} ]$ \\\hline
    5 & & & & $[720 ,\,720 ,\,{1} ]$ & $[5040 ,\,5039 ,\,{2} ]$ \\\hline
    6 & & & & & $[5040 ,\,5040 ,\,{1} ]$ \\ \hline\hline
\end{tabular}
    \caption{Parameters of the codes $\coxeter{A_m}{}{r}$. Here and below, the distance values shown in italic rely on the validity of \cref{conj: distance}.}
    \label{tab: Am params}
\end{table}

\subsection{Codes of type \texorpdfstring{$I_2(3)^\mu$}{I23m}}\label{sec: I23}
Consider Coxeter codes corresponding to the infinite family $I_2(3)^\mu$, $\mu$ copies of the order-6 dihedral group. Note that the rank of $I_2(3)^\mu$ is $m=2\mu$.

\section{Concluding remarks}
\subsection{Distance proof.}
An obvious open direction of our work is \cref{conj: distance} on the distance of a Coxeter code. In \cref{thm: exponential lower bound} we proved that the distance of the order-$r$ code of any rank-$m$ Coxeter system is $\geq2^{m-r}$. To do so, we fixed a value of $r$ and argued by induction on $m\geq r$, showing that for any non-trivial codeword in a rank $m+1$ code, there are at least two disjoint rank-$m$ standard cosets on which the codeword is supported. One route toward proving the distance conjecture is by determining a more precise lower bound on the number, $\ell$, of disjoint rank-$m$ standard cosets supporting the codeword. If, for instance, $\ell$ satisfies
\begin{equation*}
    \min_{\substack{J\subseteq S \\ \abs{J}=m-r}} \abs{\standard{J}} = \ell\cdot\!\!\!\!\min_{\substack{J\subseteq S \\ \abs{J}=m-r-1}} \abs{\standard{J}},
\end{equation*}
then \cref{conj: distance} would hold by induction.

\begin{table}[h!]
    \centering
       \begin{tabular}{L{0.05\textwidth} c c c c c}
    \hline
    & $\mu$ &  & &  &  \\ \cline{2-6}
    $r$ & 1 & 2 & 3 & 4 & 5 \\ \hline\hline
    1 & $[ 6,\, 5,\, 2]$ & $[36 ,\,9 ,\, 12]$ & $[216 ,\,13 ,\, {72}]$ & $[1296 ,\, 17,\, {432}]$ & $[ 7776,\, 21,\, {2592}]$ \\\hline
    2 & $[6 ,\,6 ,\, 1]$ & $[36 ,\,27 ,\,4]$ & $[216 ,\,64 ,\, \textit{24}]$ & $[ 1296,\, 117,\,\textit{144}]$ & $[7776 ,\, 186,\, \textit{864}]$ \\\hline
    3 & & $[36 ,\, 35,\, 2]$ & $[216 ,\,152 ,\, {8}]$ & $[ 1296,\,421 ,\, \textit{48}]$ & $[ 7776,\, 906,\, \textit{288}]$ \\\hline
    4 & &$[36 ,\, 36,\, 1]$ & $[216 ,\,203 ,\, {4}]$ & $[1296 ,\,875 ,\,{16}]$ & $[7776 ,\, 2676,\, \textit{96}]$ \\\hline
    5 & & & $[216 ,\,215 ,\,2]$& $[ 1296,\, 1179 ,\, {8}]$ & $[7776 ,\, 5100,\,{32}]$ \\\hline
    6 & & & $[216 ,\,216 ,\,1]$&$[1296 ,\, 1279,\, {4}]$ & $[7776 ,\,6870 ,\, {16}]$ \\ \hline
    7 & & & & $[ 1296,\, 1295,\, 2]$ & $[7776 ,\, 7590,\,{8}]$ \\\hline
    8 & & & & $[ 1296,\, 1296,\, 1]$& $[7776 ,\,7755 ,\, {4}]$ \\ \hline
    9 & & & &  & $[7776 ,\, 7775,\,2]$ \\\hline
    10 & & & & & $[7776 ,\, 7776,\, 1]$ \\ \hline
\end{tabular}
    \caption{Parameters of the codes $\coxeter{I_2(3)^\mu}{}{r}$.}
    \label{tab: I2(3) params}
\end{table}

\subsection{Codes of type \texorpdfstring{$I_2(4)^\mu$}{I24m}}
Consider Coxeter codes corresponding to the infinite family $I_2(4)^\mu$, $m$ copies of the order-8 dihedral group. Note that the rank of $I_2(4)^\mu$ is $2\mu$.
\begin{table}[h!]
    \centering

    \begin{tabular}{L{0.05\textwidth} c c c c}
    \hline
    &$\mu$  &  & &   \\ \cline{2-5}
    $r$ & 1 & 2 & 3 & 4  \\ \hline\hline
    1 & $[ 8,\, 7,\, 2]$ & $[ 64,\, 13,\, {16}]$ & $[512 ,\, 19,\, {128}]$ & $[4096 ,\, 25,\, {1024}]$  \\\hline
    2 & $[8 ,\, 8,\, 1]$ & $[64 ,\, 51,\,{4}]$ & $[512 ,\, 130,\, \emph{32}]$ & $[4096 ,\, 245,\,\emph{256}]$  \\\hline
    3 & & $[64 ,\, 63,\, 2]$ & $[ 512,\, 382,\, 8]$ & $[ 4096,\, 1181,\, \emph{64}]$  \\\hline
    4 & & $[64 ,\, 64,\, 1]$& $[512 ,\, 493,\,{4}]$ & $[4096 ,\, 2915,\,{16}]$  \\\hline
    5 & & & $[512 ,\, 511,\,2]$ & $[ 4096,\, 3851,\, {8}]$  \\\hline
    6 & &  & $[512 ,\, 512,\, 1]$ & $[4096 ,\, 4071,\, {4}]$  \\\hline
    7 & & & & $[4096 ,\, 4095,\,2]$  \\\hline
    8 & & & & $[4096 ,\, 4096,\, 1]$  \\\hline
\end{tabular}
    \caption{Parameters of the codes  $\coxeter{I_2(4)^\mu}{}{r}$.}
    \label{tab: I2(4) params}
\end{table}

\subsection{Further combinatorial properties.}
We have introduced a broad family of binary codes that generalizes the classic Reed--Muller family and shares several of its key features. It is natural to wonder what other properties of RM codes are shared with the Coxeter code family beyond our conjectured value of the distance. For instance, what is the equivalent notion of a \emph{projective} RM code for Coxeter codes? The codewords of minimum weight in RM codes are given by flats in the affine geometry; is there a geometric characterization of the minimum weight codewords for arbitrary Coxeter codes, and what kind of geometry could be involved?

Another line of thought is related to further combinatorial properties of Coxeter complexes, involving {\em residues} and {\em $f$-vectors} \cite{petersen2015eulerian}. We had initially phrased some of our definitions and proofs to involve these concepts before arriving at simpler arguments given here. At the same time, they may still find uses in uncovering further interesting properties of Coxeter codes and related code families.

\subsection{Local testability.} RM codes are known to have the local testability property \cite{alon2005testing}: simply check the parity of a random dual codeword of minimum weight. Supposing that their minimum weight codewords can be characterized, does the analogous local tester work for Coxeter codes? Coxeter codes are also related to codes on simplicial complexes, some of which have led to constructions of LTCs (for instance, the codes of \cite{dinur2023new}). In particular, the poset of all standard cosets of $(W,S)$, ordered by reverse inclusion, forms a simplicial complex known as the \emph{Coxeter complex}. By placing bits on the simplices of the highest dimension, the order-$r$ Coxeter code has parity checks given by $(m-r-2)$-simplices. Is there a unifying framework connecting the local testability of such simplicial codes to that of RM codes? 

\subsection{Achieving capacity and automorphisms.}
Switching to a probabilistic view, one could also study the capacity-achieving properties of Coxeter codes, extending the results for RM codes \cite{Kudekar2015ReedMullerCA}, \cite{reeves2023reed}, \cite{AbbeSandon2023}. For the binary erasure channel, it suffices to exhibit a doubly transitive action by the automorphism group of the code \cite{Kudekar2015ReedMullerCA}, and while the group $W$ naturally acts on the code space (\cref{thm: group code}), this action is only singly transitive. The automorphism group of an RM code (supposing $r\notin\br{-1,0,m-1,m}$) is given by the affine group ${Aut}(RM(r,m))=\ZZ_2^m\rtimes GL(m,2)$, far larger than 
simply $\ZZ_2^m$. Is there a suitable generalization of the affine group that captures the automorphisms of a Coxeter code?

By computer, we found that $|{Aut}(\coxeter{A_3}{}{1})|=196608=3\cdot 2^{16}$. 
This group is formed as a semi-direct product of the automorphisms of the Cayley graph of $A_3$ (given by $A3\times A1$) together with the group generated by symmetries swapping each of the 12 pairs of opposite (same-color) edges in the 6 squares of the graph; see \cref{fig: permutohedron}. This group acts transitively on the set of coordinates, but (again by computer) is not doubly transitive. Uncovering the structure
of the group $Aut(\coxeter{A_m}{}{r})$ for arbitrary $m,r$ is an interesting question, which appears nontrivial and which may elucidate the structure of $Aut(\coxeter{W}{}{r})$ in general.

\subsection{Decoding algorithms}
The accumulated lore of RM decoding comprises a vast body of results \cite{abbe2023reed}. An algorithm that is attuned to
our extension of the RM code family is {\em Recursive Projection Aggregation}, or RPA, suggested in \cite{YeAbbe2020}. Given a vector $y\in\FF_2^{2^m}$ received from the channel, decoding proceeds recursively by reducing the decoding task to
several decoding instances of codes of length $2^{m-1}$ and aggregating the obtained results by a majority decision.
Each of the shorter codes is obtained as a ``projection'' of $RM(r,m)$ on a one-dimensional subspace $\langle x\rangle$
and its cosets in $\FF_2^{2^m}$, so there are $2^m-1$ distinct instances of decoding. 

This procedure applies to the codes $\Coxeter r$, where we project the code on standard subgroups of rank 1 and their cosets. The authors of \cite{YeAbbe2020} consider this option in Sec.2 of their paper, where instead of all the subspaces,
they limit the procedure to the $m$ subspaces generated by the standard basis vectors. We leave a detailed analysis of this decoding for Coxeter codes for future work.

\subsection{Generalizing to achieve better parameters}
A major drawback of Coxeter codes is that they seemingly have worse parameters than RM codes for any given rank, $m$. In particular, the distance of high-order Coxeter codes is always equal to $2^{m-r}$ (\cref{cor: distance for large r}), whereas the code length grows much faster than $2^m$ for most Coxeter codes aside from RM codes. The poor distance occurs because with high-order codes, one can always find $m-r$ commuting generators in $(W,S)$, which form $(m-r)$-cubes. Generalizations of Coxeter codes could avoid this problem. We will mention two broad generalizations here, though we have not examined their viability in providing better parameters.

\subsubsection{Sets of generators} The first generalization is to restrict the possible choices of standard cosets.
\begin{definition}
    Let $(W,S)$ be a rank-$m$ Coxeter system, and consider some collection $\mcS\subseteq\mcP(S)$ of subsets of generators. The order-$r$ Coxeter code of type $(W,\mcS)$ is defined as
    $$
    \mcC_{(W,\mcS)}({r}) \coloneqq\Span\br{\indicator{\sigma\standard{\bigcup_{J\in\mcJ} J}}\mid \sigma\in W, \mcJ\subseteq \mcS, \abs{\mcS}=m-r}.
    $$
\end{definition}
If the collection $\mcS$ is chosen to be the collection of singletons $\mcS=\br{\br{s_i}\mid i\in[m]}$, then we recover the standard definition of a Coxeter code.

\subsubsection{Group codes}
The following is an extremely broad way to construct group codes, which has likely been studied in various capacities.
\begin{definition}
    Let $G$ be a finite group generated by a subset of $m$ elements $S\subseteq G$, i.e., $G=\standard{S}$. The order-$r$ group code of type $(G,S)$ is a left ideal the group algebra $\FF G\coloneqq \br{f\colon G\rightarrow\FF}$, defined as
    $$
    C_{(G,S)}(r) \coloneqq\Span\br{\indicator{g\standard{J}}\mid g\in G, J\subseteq S, \abs{J}=m-r}.
    $$
\end{definition}
Given a group $G$, one can prove using standard results in group theory that each choice of generating set $S$ gives a \emph{filtration} of the group algebra $\FF G$, i.e.,
$$
\br{0} =  C_{(G,S)}(-1) \;\subseteq\; C_{(G,S)}(0) \;\subseteq \;\cdots\;\subseteq\; C_{(G,S)}(m-1) \;\subseteq\; C_{(G,S)}(m) =\FF G,
$$
satisfying the multiplication property $C_{(G,S)}(r_1)\odot C_{(G,S)}(r_2)\subseteq C_{(G,S)}(r_1+r_2)$. If this generating set contains only even-order elements, then $C_{(G,S)}(r)\subseteq C_{(G,S)}(m-r-1)^\perp$, with equality likely depending on the particular combinatorial structure of the group.

A poor feature of all Coxeter codes is that for any family of Coxeter systems with increasing rank, $\br{(W_m,S_m)\mid \abs{S_m}=m}_{m\ge 1}$, the group order scales \emph{exponentially} in the rank, $\abs{W_m}=\Omega(2^m)$. That is, from a finite-scale perspective, the length of Coxeter codes grows quickly out of control. A promising direction toward constructing families of shorter codes would be to consider group codes corresponding to a family of finite groups  with explicit generating sets $(G_i,S_i)$ for which the number of group elements (the code length) grows polynomially with the number of generators
$\abs{G_i}=\poly(\abs{S_i})$. 

\section*{Acknowledgment}
We are grateful to Madhura Pathegama for helpful discussions concerning the distance estimates of Coxeter codes.

\end{document}